\documentclass[12pt]{article}
\usepackage{latexsym,amssymb,amsmath,amsfonts,amsthm,graphicx}
\newtheorem{theorem}{Theorem}

\newtheorem{lemma}{Lemma}

\theoremstyle{definition}

\def\beq{ \begin{equation} }
\def\eeq{ \end{equation} }

\def\square{\vcenter{\vbox{\hrule height .4pt
  \hbox{\vrule width .4pt height 5pt \kern 5pt
        \vrule width .4pt} \hrule height .4pt}}}

\usepackage{lineno}

\usepackage[sorting = nyt]{biblatex}
\begin{document}


\title{Competitive exclusion in a model with seasonality: three species cannot coexist in an ecosystem with two seasons}
\author{Hwai-Ray Tung and Rick Durrett \\Dept. of Mathematics, Duke University }

\date{\today}						

\maketitle

\begin{abstract} 
Chan, Durrett, and Lanchier introduced  a multitype contact process with temporal heterogeneity involving two species competing for space on the d-dimensional integer lattice. Time is divided into two seasons. They proved that there is an open set of the parameters for which both species can coexist when their dispersal range is sufficiently large. Numerical simulations suggested that three species can coexist in the presence of two seasons. The main point of this paper is to prove that this conjecture is incorrect. To do this we prove results for a more general ODE model and contrast its behavior with other related systems that have been studied in order to understand the competitive exclusion principle.
\end{abstract}

\noindent Keywords: Competitive Exclusion, Resource Competition, Periodic Environment, Dynamical Systems

\section{Introduction}
Understanding the conditions that allow for multiple species to coexist has been of longstanding interest. The competitive exclusion principle, sometimes called Gause's principle, states that $n$ resources can support at most $n$ species. For example, in Gause (1932)'s experiments with Paramecium, there was one resource, food, and the species that better utilized the food Gause gave them drove the others to extinction. However, in other situations, what constitutes a resource is not always clear. Hutchinson (1961) drew attention to this through the "Paradox of the Plankton," the enormous diversity of phytoplankton coexisting despite the small number of resources in ocean water. Many explanations for the seeming failure of the competitive exclusion principle have been explored in math models; see Armstrong and McGehee (1980) for ODE models and Hening and Nguyen (2020) for SDE and piecewise deterministic Markov process models. Hutchinson's explanation was a changing environment; times when different species are favored would be considered different niches.  

To demonstrate how temporal heterogeneity could encourage coexistence, Armstrong and McGehee (1976) considered a simple $n$ season system where $n$ species could survive on one resource. We define a season as an interval of time under which the parameters do not exhibit explicit time dependence. The system is

\begin{equation*}
    \frac{1}{n_i}\frac{dn_i}{dt} = \gamma_iRg_i(t) - \sigma_i, \qquad R = R_{max} - \sum_{i=1}^k s_in_i
\end{equation*}

\noindent where $R$ represents available resource, $k$ is the number of species, and $g_i(t)$ is a function of period $T$ that is equal to $1$ on the interval $[a_i, b_i]$ and is equal to $0$ on the intervals $[0, a_i]$ and $[b_i, T]$. When $g_i(t) = 1$ and there is no temporal heterogeneity, we recover Volterra (1928)'s model, one of the earliest models for justifying the competitive exclusion principle; the species with the highest $R_{max} - \sigma_i/\gamma_i$ wins. When there is temporal heterogeneity, coexistence becomes possible. Intuitively, $g_i(t)$ indicates whether species $i$ is in a growing season or declining season. By having disjoint growing seasons, one species would quickly grow while the others would quickly shrink, preventing them from effectively competing with the currently growing species. Armstrong and McGehee then constructively proved that in their model, parameters could be found that allowed $n$ species to coexist given $n$ seasons. Coexistence here is an example of the storage effect proposed in Chesson (1994), which outlines how species-specific responses to the environment, covariance between environment and competition, and buffered population growth can contribute to coexistence. The name of the storage effect comes from how ``storing" more benefits of advantageous times than is ``spent" during disadvantageous times can enable coexistence.

Chan, Durrett, and Lanchier (2009) considered a two-type contact process on a square lattice with long range interaction and showed that for an open set of parameters, two species can coexist in a model with two seasons. Their system is a stochastic spatial analog of
\begin{equation}
    \frac{1}{n_i}\frac{dn_i}{dt} = \gamma_i(t) R - \sigma_i \qquad R = 1 - \sum_{i=1}^k n_i
        \label{eq:cdl_mean_field}
\end{equation}
where the $\gamma_i$ are periodic functions, $R$ represents available space, and $k$ is the number of species. There is one resource $R$ so in the temporally homogeneous case one species will competitively exclude the others. In the case that the $\gamma_i(t)$ are constant on $[0,T_1]$, on $[T_1,T_2]$, and periodic, there are two seasons and therefore two niches; so, it is not surprising that two species can coexist. They speculated that a fast dispersing species could exploit the early part of a season before losing to a superior competitor, allowing for three or more species to coexist. Here, we will prove that this is not possible in the ODE.

The two-species system in CDL is a special case of the two-species periodic Lotka-Volterra model whose population sizes $n_1$ and $n_2$ are described by

\begin{align*}
    \frac{1}{n_1}\frac{dn_1}{dt} &= b_1(t)-a_{11}(t)n_1 - a_{12}(t)n_2 \\
    \frac{1}{n_2}\frac{dn_2}{dt} &= b_2(t)-a_{21}(t)n_1 - a_{22}(t)n_2
\end{align*}

\noindent 
where $b_i(t)$ and $a_{ij}(t)$ are periodic functions with period $T$. In the case that $a_{2i} = k a_{1i}$, the Lotka-Volterra model can be written as a periodic version of Volterra (1928)'s model. Cushing (1980) studied the stability of periodic solutions by generalizing the bifurcation diagrams for the constant coefficient Lotka-Volterra model, and gave an example of when there is coexistence in the periodic Lotka-Volterra model, but one of the species goes extinct when temporal variation is removed by replacing the periodic parameters with their average. Mottoni and Schiaffino (1981) study the same model using a geometric approach and, in addition to recovering some of Cushing's results, also prove that any solution approaches a solution with period $T$.

In this paper, we consider a system that we call the three-species periodic Volterra model

\begin{equation}
    \frac{1}{n_i} \frac{dn_i}{dt} = \gamma_i(t) R(n_1, n_2, n_3, t) - \sigma_i(t), \qquad i=1, 2, 3
    \label{eq:main_model}
\end{equation}

\noindent where $n_i$ is the population size of species $i$, $\gamma_i$ is the growth rate gained per available resource amount for species $i$, $R$ is the amount of available resource, and $\sigma_i$ is the death rate of species $i$.  $R, \gamma_i, $ and $\sigma_i$ are all periodic in $t$ with period $T$. To prove results about this system, we suppose that

\begin{itemize}
    \item[\textbf{A1}] $R$ is strictly decreasing with respect to population sizes $n_1, n_2,$ and $n_3$.
    \item[\textbf{A2}] $R \leq 0$ when the population size of any one species is sufficiently large.
    \item[\textbf{A3}] $\gamma_i(t)$ and $\sigma_i(t)$ are positive and upper bounded.
    \item[\textbf{A4}] $R$ is continuous with respect to $n_1, n_2,$ and $n_3$.
    \item[\textbf{A5}] We have existence and uniqueness of solutions.
    \label{model_assumptions}
\end{itemize}

\noindent
$A1-3$ are reasonable biologically. $A1$ states that a larger population means more resource consumption, and therefore less available resource. $A2$ implies that there is a limited amount of resources that cannot support infinitely large populations. $A3$ ensures that our birth and death functions have the proper sign and do not blow up. $A4-A5$ are reasonable mathematically. The system \eqref{eq:cdl_mean_field} with three species satisfies these conditions.

We also will not consider the case that there is a nontrivial triple $c_i$ such that 

$$c_1 \gamma_1 + c_2 \gamma_2 + c_3\gamma_3 = \int_0^T c_1 \sigma_1 + c_2 \sigma_2 + c_3\sigma_3 dt = 0$$

\noindent This case is also ignored when examining the competitive exclusion principle for the Volterra model with multiple resources - see page 47 of Hofbauer and Sigmund 1998. The reason is that this case represents a degenerate case where one of the species populations can be written as a function that is increasing with respect to the other two and is periodic in $t$ with period $T$. This means the system can be reduced to a two-species model. While coexistence is possible with two seasons under this case, its equilibrium lacks stability and it loses its coexistence with the slightest perturbations in $\gamma_i$ or $\sigma_i$.  

There are many different definitions of coexistence. For this paper, we say that the system exhibits coexistence if none of the species go extinct for any positive initial condition. A species goes extinct if $\lim_{t\rightarrow \infty} n_i(t) = 0$.

We show the following theorems.

\begin{theorem}
If the growth per resource rates $\gamma_i$ are linearly dependent, then the three-species periodic Volterra model does not exhibit coexistence.
\label{coro:LD}
\end{theorem}

The linear dependence assumption holds in the piecewise constant three-species model of Chan, Durrett, and Lanchier and implies that the system does not exhibit coexistence. Miller and Klausmeier (2017) also come to the same conclusion, although their arguments are not rigorous. 

Exact linear dependence is a strong condition, but our result is robust to slight deviations; we extend Theorem \ref{coro:LD} to the situation in which the $\gamma_i$ are nearly linearly dependent

\begin{theorem}\label{coro:almostLD}
Given $c_i$ not all $0$, $\sigma_i$, and $R$ for the three-species periodic Volterra model, there exists an $\epsilon>0$ such that if
\begin{equation*}
\int_0^T |c_1\gamma_1 + c_2\gamma_2 + c_3\gamma_3| dt < \epsilon
\end{equation*}
\noindent Then the model does not exhibit coexistence.
\end{theorem}

\noindent Section \ref{sec:3species} gives an important lemma used to prove the two theorems. Section \ref{sec:apps} proves the theorems and gives an example application.

\section{Condition for Coexistence and Extinction}
\label{sec:3species}
To determine if a species goes extinct, i.e., $\lim_{t\rightarrow \infty}n_i(t) = 0$, we first focus on $n_1^{c_1}n_2^{c_2}n_3^{c_3}$. This function has been used to prove results on coexistence in other models (Volterra 1928, Hofbauer 1981, Hofbauer and Sigmund 1998, Schreiber et al. 2011) and acts as an ``average Lyapunov" function whose decrease implies average movement towards faces with $c_i > 0$ and away from faces with $c_i < 0$.

Multiplying through \eqref{eq:main_model} by $c_i$, summing, and setting $\pmb{c\cdot \gamma} = c_1\gamma_1 + c_2\gamma_2 + c_3\gamma_3$ and $\pmb{c\cdot \sigma} = c_1\sigma_1 + c_2\sigma_2 + c_3\sigma_3$, we get

\begin{equation}
    \frac{1}{n_1^{c_1}n_2^{c_2}n_3^{c_3}} \frac{dn_1^{c_1}n_2^{c_2}n_3^{c_3}}{dt} = (\pmb{c\cdot \gamma}) R(n_1, n_2, n_3, t) - (\pmb{c\cdot \sigma} )
    \label{eq:c1c2c3}
\end{equation}

\noindent For some systems, an appropriate choice of $c_1, c_2,$ and $c_3$ will let us ignore $R$ and show that $n_1^{c_1}n_2^{c_2}n_3^{c_3} \rightarrow 0$. Once this is established, we can use the following lemma to preclude coexistence.

\begin{lemma}
The three-species periodic Volterra model \eqref{eq:main_model} does not exhibit coexistence iff there exist constants $c_1, c_2,$ and $c_3$ that are not all positive and
$$\lim_{t\rightarrow \infty}n_1^{c_1}n_2^{c_2}n_3^{c_3} = 0.$$
\label{thm:main}
\end{lemma}

The remainder of this section describes the ideas behind the proof of Lemma \ref{thm:main}. We start with the easier direction. If the three-species periodic Volterra model does not exhibit coexistence, then there exists a species, which we label as species 1, whose population approaches $0$. Setting $c_1=1$ and $c_2=c_3=0$ completes this direction.
 
We now proceed with the other direction by considering the possible cases of the signs of $c_i$. Recalling that $n_i(t)$ is upper bounded by A2, if $c_i$ is nonpositive then $n_i^{c_i}$ is bounded from below. This implies that for case 1, where all $c_i$ are nonpositive, then $n_1^{c_1}n_2^{c_2}n_3^{c_3}$ cannot approach $0$ and we can ignore this case. This also implies that for case 2, where only one of the $c_i$ is positive, which we label as species $1$, then $n_1^{c_1}n_2^{c_2}n_3^{c_3} \rightarrow 0$ implies $n_1^{c_1} \rightarrow 0$ as $t \rightarrow \infty$ and therefore that species $1$ goes extinct.

Case 3, where two of the $c_i$ are positive and one is nonpositive, is more involved. Let $c_1, c_2 > 0$ and $c_3 < 0$. Using that $n_i$ is upper bounded once more,

\begin{equation}
    n_1(t)^{c_1}n_2(t)^{c_2} \rightarrow 0
    \label{eq:case4ineq}
\end{equation}

\noindent In order to show that species $1$ or $2$ goes extinct, we need to rule out the possibility that species $1$ and $2$ take turns approaching $0$, keeping $\limsup n_1 = n_1^*$ and $\limsup n_2 = n_2^*$ positive. To do so, we note that by \eqref{eq:case4ineq}, paths from $(n_1^*, 0)$ to $(0, n_2^*)$ must travel near the origin after some time. Then, we show that if the trajectory of $(n_1, n_2)$ nears the origin and eventually leaves, then $(n_1, n_2)$ will consistently leave the origin in the same direction, without loss of generality towards $(n_1^*, 0)$. This implies that $n_2^* = 0$ and therefore species $2$ goes extinct. The details of the proof of case 3, as well as a visual overview of the proof, can be found in Section \ref{sec:proof}.

\section{Applications}
\label{sec:apps}
In this section, we use Lemma \ref{thm:main} to prove Theorems \ref{coro:LD} and \ref{coro:almostLD} and give examples of systems where we can rule out coexistence.

\begin{proof}[\textbf{Proof of Theorem \ref{coro:LD}}]
Since the $\gamma_i$ are linearly dependent, we can find $c_i$ such that $\pmb{c \cdot \gamma} = 0$. Then by \eqref{eq:c1c2c3},

$$\frac{d}{dt}\left[\ln\left(n_1^{c_1}n_2^{c_2}n_3^{c_3}\right)\right] = - \pmb{c\cdot \sigma}$$

\noindent Flipping the signs of $c_i$ if necessary, we can assume without loss of generality that $\pmb{c\cdot \sigma}>0$. This implies $n_1^{c_1}n_2^{c_2}n_3^{c_3} \rightarrow 0$, so applying Lemma \ref{thm:main} completes the proof.
\end{proof}

\noindent \emph{\textbf{Application of Theorem \ref{coro:LD}: Contact process with seasons.}}
One notable example of a model where the $\gamma_i$ are linearly dependent is the mean field limit of the three-species two-seasons model that appears in Chan, Durrett, and Lanchier (2009) - see \eqref{eq:cdl_mean_field}. They showed that, in the absence of species 3, coexistence occurs when 

$$\frac{1}{T}\int_0^T \gamma_1(1-\overline{n}_2) - \sigma_1 dt > 0 \text{ and } \frac{1}{T}\int_0^T \gamma_2(1-\overline{n}_1) - \sigma_2 dt > 0$$

\noindent where $\overline{n}_i$ is the nontrivial periodic solution to $\frac{1}{n_i}\frac{dn_i}{dt} = \gamma_i(t) (1-n_i) - \sigma_i$. The first integral represents the net growth rate of species $1$ when $n_1$ has been small for a long time, giving $n_2$ time to converge to $\overline{n}_2$. If the growth rate is positive, then species 1 won't go extinct. Similarly, the second condition represents that species $2$ has a positive growth rate when $n_2$ is small and $n_1$ is near $\overline{n}_1$. Using the ODE result, they showed that the same conditions guaranteed coexistence for the two-type contact process on the square lattice with long range interactions. 

Chan, Durrett, and Lanchier also conjectured that three species could coexist with two seasons. This could be true in their stochastic model, but it does not hold in the mean field limit. To prove this, since there are only two seasons and $\gamma_i$ is a function of the season, the space of possible $\gamma_i$ has dimension $2$. There are three species, so the $\gamma_i$ are linearly dependent. Therefore, by Theorem \ref{coro:LD} the three species cannot coexist.

\begin{proof}[\textbf{Proof of Theorem \ref{coro:almostLD}}]
In the case where $c_i$ are not all the same sign, note that $|R(n_1, n_2, n_3, t)|$ is bounded since $R$ has monotonicity and $n_i$ is bounded. Integrating \eqref{eq:c1c2c3}, we get

\begin{align}
\begin{split}
    \ln\left[\frac{n_1(t+T)^{c_1}n_2(t+T)^{c_2}n_3(t+T)^{c_3}}{n_1(t)^{c_1}n_2(t)^{c_2}n_3(t)^{c_3}}\right] &= \int_t^{t+T}(\pmb{c \cdot \gamma})R - (\pmb{c \cdot \sigma}) ds \\
    &\leq \int_0^T |\pmb{c \cdot \gamma}|\max|R| - (\pmb{c \cdot \sigma}) ds
    \label{eq:integrated}
\end{split}
\end{align}

\noindent Then, flipping the signs of $c_i$ if necessary, setting 

$$\epsilon < \frac{1}{\max|R(n_1, n_2, n_3, t)|}\int_0^T \pmb{c \cdot \sigma} ds$$

\noindent will force $n_1(t+T)^{c_1}n_2(t+T)^{c_2}n_3(t+T)^{c_3} < n_1(t)^{c_1}n_2(t)^{c_2}n_3(t)^{c_3}$, and therefore $n_1^{c_1}n_2^{c_2}n_3^{c_3} \rightarrow 0$. Applying Lemma \ref{thm:main} completes the proof.

In the case where $c_i$ all have the same sign, we assume without loss of generality that $c_i$ are all positive. Then,

\begin{align}
\begin{split}
    \ln\left[\frac{n_1(t+T)^{c_1}}{n_1(t)^{c_1}}\right] &= \int_t^{t+T}c_1\gamma_1 R - c_1\sigma_1 ds \\
    &\leq \int_0^T (\pmb{c \cdot \gamma})\max|R| - c_1\sigma_1 ds 
\end{split}
\end{align}

\noindent Setting

$$\epsilon < \frac{1}{\max|R(n_1, n_2, n_3, t)|}\int_0^T c_1 \sigma_1 ds$$

\noindent forces $n_1(t) \rightarrow 0$ as $t \rightarrow \infty$, which implies extinction of species $1$ and no coexistence.  

\end{proof}

\newpage

\noindent \emph{\textbf{Application of Theorem \ref{coro:almostLD}: Numerical example.}} To conclude this section, we do a concrete example. Consider the system

\begin{align*}
    \frac{1}{n_1}\frac{dn_1}{dt} &= \gamma(3, 5, t)R(n_1, n_2, n_3) - 1 \\
    \frac{1}{n_2}\frac{dn_2}{dt} &= \gamma(4.5, 3.4, t)R(n_1, n_2, n_3) - 1 \\
    \frac{1}{n_3}\frac{dn_3}{dt} &= \gamma(4.1, 3.78, t)R(n_1, n_2, n_3) - 1 \\
    R &= 1-n_1-n_2-n_3 \\
    \gamma(a, b, t) &=  \begin{cases} 
      a & 0 < t \leq 0.6 \\
      a(b/a)^{(t-0.6)/0.4} & 0.6 < t \leq 1 \\
      b & 1 < t \leq 1.6 \\
      b(a/b)^{(t-1.6)/0.4} & 1.6 < t \leq 2 \\
      \end{cases}
\end{align*}

\noindent The system can be viewed as an extension of the two-season model - see Fig \ref{fig:example}A; instead of making $\gamma_i$ piecewise constant, $\gamma_i$ now has a transition period between the two seasons, making the $\gamma_i$ linearly independent. However, the $\gamma_i$ are close enough to being linearly dependent that we can preclude coexistence.

To show that the  cannot coexist, we mimic the proof of Theorem \ref{coro:almostLD}. We first bound $R$ from above. Note that $\gamma_i \geq 3$. This implies that when $R > 1/3$, then $dn_i/dt > 0$ and therefore $dR/dt < 0$. Thus, after some time, $R \leq 1/3$. Next, we set $c_1 = -1, c_2 = -916/307,$ and $c_3 = 1230/307$; these were chosen to make the $\gamma_i$ linearly dependent during times $[0, 0.6] \cup [1, 1.6]$. Now, we integrate.

$$\int_0^T \pmb{c \cdot \sigma}dt = \frac{14}{307} \approx 0.046 $$
$$\max|R(n_1, n_2, n_3, t)|\int_0^T |\pmb{c \cdot \gamma}| dt \leq 0.041 $$

\noindent Since $0.041 < 0.046$, by \eqref{eq:integrated}, $n_1^{c_1}n_2^{c_2}n_3^{c_3} \rightarrow 0$. Applying Lemma \ref{thm:main} implies that one of the species goes extinct. This can be seen in Figure \ref{fig:example}.

\begin{figure}
    \centering
    \includegraphics[scale = 0.6]{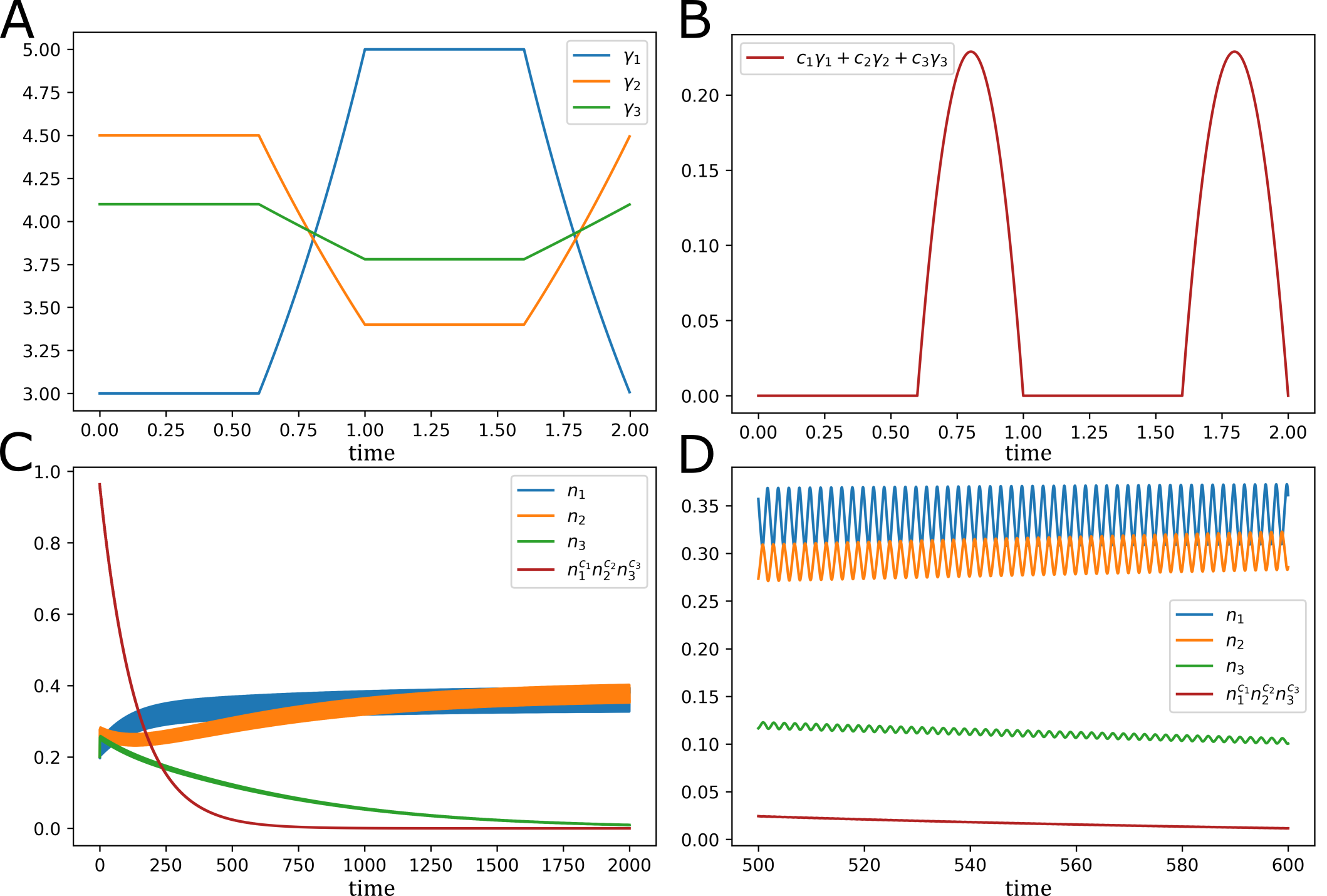}
    \caption{The graphs are based on the example system given in Section \ref{sec:apps}. A) One period of growth functions $\gamma_i$. Instead of the two-season model considered in Theorem \ref{coro:LD}, we add transition seasons to make the $\gamma_i$ continuous. B) Measure of linear dependence $\pmb{c \cdot \gamma}$ for our choice of $c_1, c_2,$ and $c_3$. Since $\pmb{c \cdot \gamma}$ is sufficiently close to $0$, we can preclude coexistence. C) Population dynamics. Species $3$ goes extinct. D) Population dynamics zoomed. The populations oscillate over time due to changes in $\gamma_i$.}
    \label{fig:example}
\end{figure}

\newpage

\section{Proof of Lemma \ref{thm:main}}
\label{sec:proof}

\begin{figure}
    \centering
    \includegraphics[scale = 0.3]{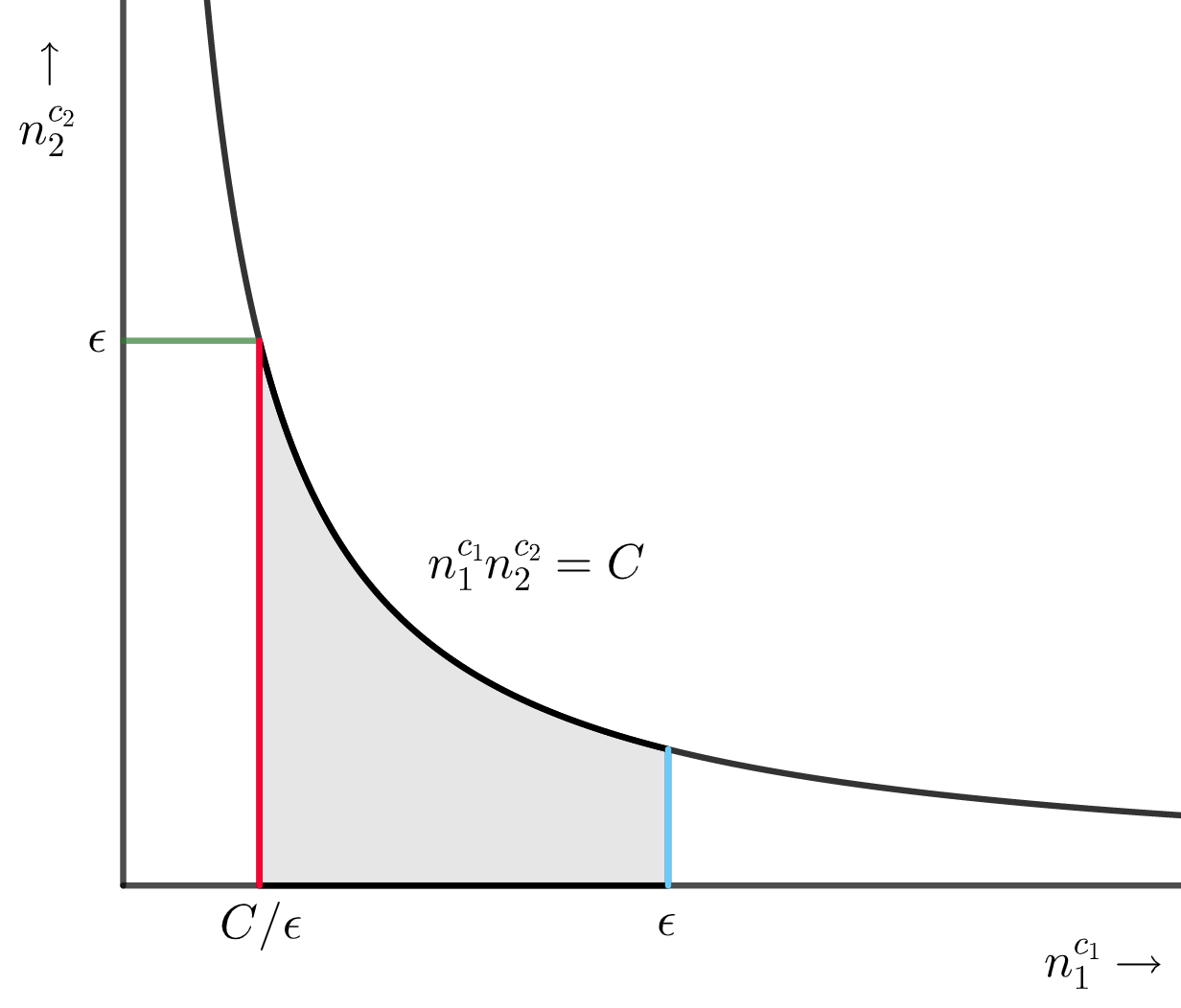}
    \caption{Visual for the proof of case three of Lemma \ref{thm:main}. From \eqref{eq:case4ineq}, after sufficiently large time, the solution must remain below the curve $n_1^{c_1}n_2^{c_2} = C$. To disprove coexistence, we need to show that $(n_1^{c_1}, n_2^{c_2})$ cannot move between $(n_1^*, 0)$ and $(0, n_2^*)$. This is equivalent to proving that the solution cannot move from the blue line to the green line, or vice versa. Lemma \ref{lem:R3} shows that spending sufficient time between the green and blue lines causes $n_1(t+T)/n_1(t)$ to converge to being positive or negative. Lemma \ref{lem:R1} shows that sufficient time for convergence is eventually always achieved.  We now have two cases. If the sign is positive, then the solution will not be able to move from blue to green; $n_1(t+T)/n_1(t)$ will become positive before reaching the red line, preventing the solution from reaching green. Similarly, if the sign is negative, then the solution will not be able to move from green to blue}
    \label{fig:case3}
\end{figure}

Here, we give the details for case 3 in the proof of Lemma \ref{thm:main}. For a visual overview of the proof, see Figure \ref{fig:case3}. We start by proving the following two lemmas.

\begin{lemma}
Let $n_1(t)^{c_1}n_2(t)^{c_2} < C$. The time needed for $n_1(t)^{c_1}$ to pass through the interval $D_1 = [C/\epsilon, \epsilon]$ approaches infinity as $C \rightarrow 0$.
\label{lem:R1}
\end{lemma}

\begin{lemma}
If $(n_1, n_2)$ passes through the region $D = \{(n_1, n_2) | 0 \leq n_1(t)^{c_1}, n_2(t)^{c_2} \leq \epsilon \}$ starting at some sufficiently large time $\tau$, then $n_3(\tau)$ is bounded from below.
\label{lem:R2}
\end{lemma}

We define passing through $D_1$ as moving from $n_1(t)^{c_1} = \epsilon$ to $n_1(t)^{c_1}= C/\epsilon$ or vice versa without leaving $D_1$ and passing through $D$ as moving from $n_1(t)^{c_1} = \epsilon$ to $n_2(t)^{c_2}= \epsilon$ or vice versa without leaving $D$. To prove Lemma \ref{lem:R1}, note that

$$c_1\max_t \left[ \gamma_1 R(0, 0, 0, t) - \sigma_1 \right] >\frac{d}{dt}\left[\ln\left(n_1^{c_1}\right)\right] > c_1\min_t \left[  \gamma_1 R(M, M, M, t) - \sigma_1 \right]$$

\noindent Letting the RHS be $p_{min}$ and LHS be $p_{max}$, the amount of time spent traveling from $\epsilon$ to $C/\epsilon$ and the other direction is lower bounded by

$$\frac{1}{p_{min}}\ln{(C/\epsilon^2)} \text{ and } \frac{1}{p_{max}}\ln{(\epsilon^2/C)}$$

\noindent respectively. As $C$ approaches 0, both expressions, and therefore the time for pass through $D_1$, approach infinity.

Now, we prove Lemma \ref{lem:R2}. By Lemma \ref{lem:R1} and \eqref{eq:case4ineq}, if $\tau$ is sufficiently large, $(n_1(\tau), n_2(\tau))$ cannot pass through $D$ by time $\tau + T$. We now aim to show that if we start on $D$ where $n_1(\tau)^{c_1} = \epsilon$, then if $n_3(\tau)$ is small, $n_1(\tau + T)^{c_1} > \epsilon$ and there is no passing through; the other side, where $n_2(\tau)^{c_2} = \epsilon$ can be proven similarly. Note that
\begin{align*}
    \ln\left[\frac{n_1(\tau + T)^{c_1}}{n_1(\tau)^{c_1}}\right] &= c_1\int_{\tau}^{\tau + T} \gamma_1 R(n_1, n_2, n_3, t) - \sigma_1 dt \\ &\geq c_1\int_{\tau}^{\tau + T} \gamma_1 R(\epsilon^{1/c_1}, \epsilon^{1/c_2}, n_3, t) - \sigma_1 dt
\end{align*}

\noindent When $\epsilon$ is sufficiently close to $0$ and $n_3 =0$, the RHS must be positive, else it would imply that species $1$ would go extinct even without competition. By continuity of $R$ (A5), there exists some constant $a > 0$ where the RHS is still positive when $n_3(\tau) \leq a$, and therefore $n_1(\tau + T) > n_1(\tau)$. This would make $n_1$ leave $D$ without passing through. As such, to pass through $D$, there is a lower bound on the population of species $3$.

In order to prove Lemma \ref{lem:R3}, we first need an understanding the dynamics of the system when only one species is present.

\begin{lemma}
When $n_1 = n_2 = 0$, there exists at most 1 nontrivial periodic orbit $n_3^*$ for species 3. If $n_3^*$ exists and we have a nontrivial solution $n_3^{**}$, then $n_3^{**} \rightarrow n_3^*$.
\label{lem:1converge}
\end{lemma}

\begin{proof}
For simplicity of notation, we write $R(0, 0, n_3, t)$ as $R(n_3, t)$. Suppose there are two nontrivial solutions $n_3^*$ and $n_3^{**}$, with $n_3^*$ being a periodic orbit. Then
\begin{align*}
    \frac{d\ln(n_3^*)}{dt} &= \gamma_3 R(n_3^*, t) - \sigma_3 \\
    \frac{d\ln(n_3^{**})}{dt} &= \gamma_3 R(n_3^{**}, t) - \sigma_3 
\end{align*}
\noindent Subtracting, we get

$$\frac{d\ln(n_3^*/n_3^{**})}{dt} = \gamma_3 [R(n_3^*, t) - R(n_3^{**}, t)]$$

\noindent WLOG $n_3^*(0) \geq n_3^{**}(0)$. By the uniqueness condition, $n_3^*(t)=n_3^{**}(t)$ for any $t$ iff $n_3^*(0) = n_3^{**}(0)$. As such, $n_3^*(t) \geq n_3^{**}(t)$, which implies $R(n_3^*, t) - R(n_3^{**}, t) \leq 0$, with equality only when $n_3^* = n_3^{**}$. 

We first establish that $n_3^*$ is a unique periodic orbit. If $n_3^{**}$ is also a periodic orbit, then

$$0 = \int_0^T \frac{d\ln(n_3^*/n_3^{**})}{dt} dt = \int_0^T \gamma_3 [R(n_3^*, t) - R(n_3^{**}, t)] dt$$

\noindent Since $\gamma_3(t) > 0$, this implies $n_3^* = n_3^{**}$. 

To address the second claim, if $n_3^{**}$ is not a periodic orbit, then note that

$$0 > \int_0^T \gamma [R(n_3^*, t) - R(n_3^{**}, t)] dt = \int_0^T \frac{d\ln(n_3^*/n_3^{**})}{dt} dt = -\ln(n_3^{**}(T)) + \ln(n_3^{**}(0))$$

\noindent As such, $n_3^{**}$ is increasing every cycle and approaches $n^*$.
\end{proof}

Having established the existence and uniqueness of an equilibrium when only one species is present, we are now ready to prove Lemma \ref{lem:R3}.

\begin{lemma}
    For sufficiently small $\epsilon$, when $(n_1, n_2)$ is passing through $D$, then after finite time $s$, 

    $$\frac{n_1(\tau + T)}{n_1(\tau)} = \int_{\tau}^{T+\tau} \gamma_1 R(n_1, n_2, n_3, t) - \sigma_1 dt \text{ and  }\int_0^T \gamma_1 R(0, 0, n^*, t) - \sigma_1 dt$$

    \noindent have the same sign, where $n^*$ is the nontrivial equilibrium solution for $n_3$ in the absence of the other two species.
    \label{lem:R3}
\end{lemma}

To prove, we first note that by monotonicity,

$$R(\epsilon^{1/c_1}, \epsilon^{1/c_2}, n_3, t) < R(n_1, n_2, n_3, t) \leq R(0, 0, n_3, t)$$

\noindent Let $R(0, 0, n_3, t) - R(\epsilon^{1/c_1}, \epsilon^{1/c_2}, n_3, t) < m$. Then by the ODE comparison theorem, we know that $n_3$ is bounded between the solutions for

$$\frac{1}{n_3}\frac{dn_3}{dt} = \gamma_3 (R(0, 0, n_3, t) - m) - \sigma_3, \quad \frac{1}{n_3}\frac{dn_3}{dt} = \gamma_3 R(0, 0, n_3, t) - \sigma_3$$

Let $n^*$ be the equilibrium solution for the upper bound and $n_m$ the solution for the lower bound. By continuity of $R$ we can find an $\epsilon$ that lets $m$ be arbitrarily small. Applying Lemma \ref{lem:1converge}, $n_m$ must approach its equilibrium. Then,

$$0 = \lim_{\tau\rightarrow \infty}\int_{\tau}^{\tau + T} \gamma_3 R(0, 0, n_m, t) - (\sigma_3 + \gamma_3m) dt =  \int_0^T \gamma_3R(0, 0, n^{*}, t) - \sigma_3$$

\noindent Rearranging the above,

$$\lim_{\tau\rightarrow \infty}\int_{\tau}^{\tau + T} \gamma_3 (R(0, 0, n_m, t) - R(0, 0, n^{*}, t)) dt = m\int_0^T \gamma_3 dt$$

\noindent which approaches $0$ as $\epsilon$ approaches $0$. Noting that $\gamma_3 > 0$ and $R(0, 0, n^*, t) < R(0, 0, n_m, t)$ implies 

$$\lim_{\epsilon \rightarrow 0} \lim_{\tau\rightarrow \infty}\int_{\tau}^{\tau + T} R(0, 0, n_m, t) - R(0, 0, n^{*}, t) dt = 0$$

\noindent and subsequently

$$\lim_{\epsilon \rightarrow 0} \lim_{\tau\rightarrow \infty}\int_{\tau}^{\tau + T} \gamma_i R(0, 0, n_m, t) = \int_0^T \gamma_i R(0, 0, n^*, t)$$

To show that the integrals have the same sign happens after time $s$ regardless of $n_3(0)$ at time of entering $D$, recall from Lemma \ref{lem:R2} that $n_3(0)$ has a nonzero lower bound $\underline{m}$ and the upper bound $M$. By monotonicity, $n_m$ will take longest to reach the same sign when $n_m(0)=M$ or $\underline{m}$. Take $s$ to be the longer time. As $n_m \leq n_3 \leq n^*$, we have our desired result. 

We are now ready to prove extinction in case 3. By Lemma \ref{lem:R3}, $n_1(t+T)/n_1(t)$ will always be positive or negative after time $s$ in $D$, which we know will happen from Lemma \ref{lem:R1}. If the sign is negative, then $n_1^{c_1}$ would shrink before reaching $\epsilon$, and therefore the solution cannot move from $(0, n_2^*)$ to $(n_1^*, 0)$. If the sign is positive, then $n_1^{c_1}$ would grow before reaching $C/\epsilon$, which implies that $n_2^{c_2} < \epsilon$ and the solution cannot move from $(n_1^*, 0)$ to $(0, n_2^*)$. As such, we have a contradiction, and the $\limsup$ of $n_1$ or $n_2$ is 0. This concludes case 3.

\section*{Acknowledgements}
The authors would like to thank the editor and both reviewers for their valuable comments. Both authors were partially supported by the National Science Foundation grant 1809967.

\nocite{*}
\printbibliography

\end{document}